\documentclass[letterpaper, 10 pt, journal, twoside]{ieeetran}

\IEEEoverridecommandlockouts                              

\usepackage{graphicx}
\usepackage[utf8]{inputenc}
\usepackage{psfrag}
\usepackage{amsthm}
\usepackage{amsmath}
\usepackage{amssymb}

\usepackage{cite}

\newcommand{\abs}[1]{\left\vert #1 \right\vert}

\newcommand{\R}{{\mathbb R}}  

\usepackage{graphicx}

\usepackage{xcolor}	

\newtheorem{theorem}{Theorem}
\newtheorem{lemma}[theorem]{Lemma}

\newtheorem{corollary}[theorem]{Corollary}
\newtheorem{definition}{Definition}
\newtheorem{remark}{Remark}

\title{Exploiting nodes symmetries to control synchronization and consensus patterns in multiagent systems}

\author{Davide Fiore, \emph{Student Member, IEEE}, Giovanni Russo, \emph{Member, IEEE}, and Mario di Bernardo, \emph{Fellow, IEEE}
\thanks{G. Russo and D. Fiore are with the Optimization and Control Group of IBM Research Ireland, IBM Technology Campus, Mulhuddart, Dublin 15, Ireland
        (email: dvd.fiore@gmail.com; grusso@ie.ibm.com).}%
\thanks{M. di Bernardo is with the Department of Electrical Engineering and Information Technology, University of Naples Federico II, Via Claudio 21, 80125 Naples, Italy, and also with the Department of Engineering Mathematics, University of Bristol, University Walk, BS8 1TR Bristol, U.K. 
        (email: mario.dibernardo@unina.it).}%
\thanks{D. Fiore started this work while being at University of Naples Federico II.}     

}

\begin{document}

\maketitle
\thispagestyle{empty}
\pagestyle{empty}

\begin{abstract}
We present new conditions to obtain synchronization and consensus patterns in complex network systems. The key idea is to exploit symmetries of the nodes' vector fields  to induce a desired synchronization/consensus pattern, where nodes are clustered in different groups each converging towards a different synchronized evolution.
We show that the new conditions we present offer a systematic methodology to design a  distributed network controller  able to drive a network of interest towards a desired synchronization/consensus pattern. 
\end{abstract}


\section{Introduction}

\IEEEPARstart{N}{etwork} control is of utmost importance in many application areas, from computer science to power engineering, the emerging ``Internet of Things'' and  computational biology \cite{Liu_Bar_Slo_11,whalen2015observability}. Over the past few years there has been considerable interest in the problem of steering the dynamics of network agents towards some coordinated collective behavior, see e.g. \cite{Che_13} and references therein. Synchronization and consensus are two examples of such collective behavior where all the agents  {\em cooperate} in order for a common asymptotic behavior to emerge \cite{Dor_Bul_14}. 

Often, in applications, interactions between neighboring network nodes are not all collaborative as there might be certain nodes that have {\em antagonistic} relationships with neighbors. This is the case, for example, of social networks, where network agents might have different opinions \cite{Gal_96}, or biochemical and gene regulatory  networks, where interactions between nodes are either activations or inhibitions \cite{Sza_Ste_Per_06}. 
A convenient way to model the presence of collaborative and antagonistic relationships among nodes in  a network is to use {\em signed graphs} \cite{Hei_46}. Motivated by applications, an increasing number of papers in the literature is focusing on the study of the collective dynamics emerging in this type of  networks. For example, in \cite{zhang2001partial} partial synchronization of R\"{o}ssler oscillators over a ring is studied via the Master Stability Function (MSF), while in \cite{pogromsky200265} the same phenomenon is studied within the broader framework of symmetries intrinsic to the network structure (see also \cite{Rus_Slo_12} for a discussion on the interplay between symmetries and synchronization). Symmetries in the network topology have also been exploited in \cite{pecora2014cluster}, where the MSF is used to study local stability of synchronized clusters of nodes. A particularly interesting problem is the one considered in \cite{altafini2013consensus}, where sufficient conditions are given for a signed network of integrators to achieve a form of {\em ``agreed upon dissensus''}. The model proposed in \cite{altafini2013consensus} has been used in a number of applications, like flocking \cite{fan2013bipartite} and extended to the case of LTI systems and time-varying topologies, see e.g. \cite{zhang2014bipartite,zhang2016bipartite,valcher2014consensus,Pro_Mat_Cao_14}. More recently, bipartite synchronization in a network of scalar nonlinear systems whose vector fields are odd functions has been studied in \cite{zhai2016pinning}.

In this paper, we focus on studying the dynamics of networks of $n$-dimensional nonlinear nodes after performing a suitable transformation of the state variables. The specific transformation depends on the symmetries available at the nodes, rather than the symmetries of the network topology, and on the specific desired synchronization/consensus pattern. We show that studying the dynamics of the network in the new state variables simplifies the stability and convergence analysis yielding a set of sufficient conditions for the onset of synchronization/consensus patterns that can be straightforwardly verified. Finally, using these conditions, we present an intuitive systematic methodology to design distributed control algorithms, which exploit the symmetries at the nodes to achieve some desired synchronization pattern. The effectiveness of the theoretical results are illustrated via a set of representative examples.

\section{Mathematical preliminaries}
\label{sec:intro}
We denote by $I_n$ the $n\times n$ identity matrix and by $O_n$ the $n\times n$ matrix with all zero elements. The orthogonal symmetry group will be denoted by $\mathbb{O}(n)$ (see e.g. \cite{golubitsky2003symmetry}).

\subsection{Networks of interest}
We consider undirected networks of $N>1$ smooth $n$-dimensional dynamical systems 
%
\begin{equation}
\label{eq:network_diffusive_nonlinear}
\dot x_i = f\left(t,x_i\right) + k\, \sum_{j=1}^N a_{ij}\,  \big(g_{ij}\left( x_j\right) - x_i\big),
\end{equation}
with initial conditions $x_{i,0} := x_i(t_0)$, $t_0 \geq 0$, where $x_i\in \mathbb{R}^n$, $i=1,\dots,N$, is the state vector of  node $i$, $f:\R^+\times\R^n\rightarrow\R^n$ describes the intrinsic dynamics all nodes share, $k>0$ is the coupling strength, $a_{ij}\in\{0,1\}$ are the elements of the adjacency matrix, the functions $g_{ij}\left(\cdot\right)$ are the coupling functions that will be designed in this paper to obtain a specific synchronization pattern (as defined in Section \ref{sec:problem_statement}). We assume that well-posedness conditions are satisfied so that a solution of \eqref{eq:network_diffusive_nonlinear} exists for all $t\geq t_0$.

Note that, if in \eqref{eq:network_diffusive_nonlinear} we set  $g_{ij}(x) = x$, $\forall i,j=1,\dots,N$, then \eqref{eq:network_diffusive_nonlinear}  describes a network of diffusively coupled nodes, whose dynamics can be written in compact form as:
\begin{equation}
\label{eq:network_block}
\dot{X}=F(t,X)- k\,\left(L\otimes I_n\right)X,
\end{equation}
where $X:=\left[x_1^T, \;\dots\;, x_N^T\right]^T\in\mathbb{R}^{nN}$,  $F(t,X):=\big[f(t,x_1)^T, \allowbreak \;\dots\;, \allowbreak f(t,x_N)^T \big]^T$, and $L$ is the $N\times N$ Laplacian matrix. 
In the rest of the paper we will refer to networks of the form \eqref{eq:network_block} as {\em auxiliary networks} associated to \eqref{eq:network_diffusive_nonlinear}. Specifically, we will provide conditions for the onset of synchronization patterns for network \eqref{eq:network_diffusive_nonlinear} which correspond to achieving synchronization of network \eqref{eq:network_block}, as defined below.
\begin{definition}
\label{def:synchronization}
Let $\dot{s}=f(t,s)$. We say that \eqref{eq:network_block} achieves synchronization if $\lim_{t \to +\infty} \abs{x_i(t)-s(t)}=0$, $\forall i= 1,\ldots,N$. 
\end{definition}
In the case where nodes’ dynamics are integrators, Definition \ref{def:synchronization}  becomes a definition for consensus.

\subsection{Equivariant dynamics}
The symmetries of a system of ODEs are described in terms of a group of linear transformations of the variables that preserves the structure of the equation and its solutions (see \cite{golubitsky2003symmetry,dionne1996coupled,golubitsky2015recent} for a detailed discussion and proofs of the material reported in this Section). In this paper, we will consider symmetries of ODEs specified in terms of compact Lie groups acting on $\mathbb{R}^n$. These groups can be identified as a subgroup of orthogonal matrices $\mathbb{O}(n)$, i.e. matrices $\gamma$ such that $\gamma^{-1} = \gamma^T$.

Consider a dynamical system of the form
\begin{equation}
\label{eq:ode}
\dot x = f(t,x), \ \ \ x \in\R^n.
\end{equation}
where $f:\R^+\times\R^n\rightarrow\R^n$ is a smooth vector field. 
We will use the following standard definitions \cite{golubitsky2003symmetry}.
\begin{definition}
\label{def:symmetric_solutions}
The group element $\gamma\in\mathbb{O}(n)$ is a \emph{symmetry} of \eqref{eq:ode} if for every solution $x(t)$ of \eqref{eq:ode}, $\gamma x(t)$ is also a solution.
\end{definition}

\begin{definition}
Let $\Gamma$ be a compact Lie group acting on $\mathbb{R}^n$. Then, $f$ is \emph{$\Gamma$-equivariant} if $f(t,\gamma x)=\gamma f(t,x)$ for all $\gamma\in\Gamma,\; x\in\mathbb{R}^n$.
\end{definition}

Essentially, $\Gamma$-equivariance means that the orthogonal matrix $\gamma$ {\em commutes} with $f$ and it implies that $\gamma$ is a symmetry of \eqref{eq:ode}. In fact, let $y(t) = \gamma x(t)$, we have that $\dot y = \gamma \dot x = \gamma f(t,x) = f(t,\gamma x) = f(t,y)$. We now introduce the following Lemma which will be used later in the paper. 
\begin{lemma}
\label{lemma:F_equivariant}
Assume that, for system \eqref{eq:ode}, $f(t,x)$ is $\Gamma$-equivariant. Let 
\begin{equation}
\label{eq:matrix_D}
D :=\mathrm{diag}\{ \sigma_1,\dots,\sigma_N \},
\end{equation}
be a block diagonal matrix with blocks $\sigma_i \in \Gamma$, $i=1,\ldots,N$. Then,  for all $X$, $DF(t,X)=F(t,DX)$. 
\end{lemma}
\begin{proof}
The proof can be immediately obtained from \cite{dionne1996coupled} and is omitted here for the sake of brevity.
%
\end{proof}
Lemma \ref{lemma:F_equivariant} implies that whenever a function $f(t,x)$ is $\Gamma$-equivariant, then the stack $F$ commutes with $D$.  
Recently,  symmetries of dynamical systems have been investigated in several application fields, ranging from chaos and bifurcation theory  to synchronization \cite{golubitsky2003symmetry}. It has been suggested that the interplay between symmetries and dynamics plays a major role  in opinion formation models and biological applications \cite{golubitsky2015recent,Rus_Slo_12,sontag2017dynamic}. A possible justification for this might be that system behaviors induced by symmetries are {\em rigid} i.e. are robust with respect to certain perturbations of the vector field, \cite{golubitsky2015recent,josic2006network}.
\section{Bipartite synchronization}
\label{sec:bipartite_synch}
\subsection{Problem Statement}
\label{sec:problem_statement}
Let $\mathcal{G}_N := \left\{1, \ldots, N \right\}$ be the set of all network nodes and let $\mathcal{G}$ and $\mathcal{G}^\ast$ be two subsets (or \emph{cluster}) such that:
$\mathcal{G} \cap \mathcal{G}^\ast = \left\{ \emptyset \right\}$, $\mathcal{G} \cup \mathcal{G}^\ast = \mathcal{G}_N$, with the cardinality of $\mathcal{G}$ being equal to $\ell$ and the cardinality of $\mathcal{G}^\ast$ being $N-\ell$. Clearly, the two sets above generate a partition of the network nodes. Throughout this paper, no hypotheses will be made on the network partition, i.e. nodes can be partitioned arbitrarily, furthermore nodes belonging to the same cluster do not necessarily need to be directly interconnected.
\begin{definition}
\label{def:bipartite_sync}
Consider network \eqref{eq:network_diffusive_nonlinear} and let $s(t)=\gamma\, s^\ast(t)$, with $\gamma\in\mathbb{O}(n)$. We say that \eqref{eq:network_diffusive_nonlinear} achieves a \emph{$\gamma$-bipartite synchronization pattern} if: (i) $\lim_{t\to +\infty} |x_i(t)-s(t)|=0$, $\forall i\in \mathcal{G}$; and (ii) $\lim_{t\to +\infty} |x_i(t)-s^\ast(t)|=0$, $\forall i\in \mathcal{G}^\ast$.
\end{definition}

Definition \ref{def:bipartite_sync} implies that the collective behavior emerging from the network dynamics will encompass two clusters of nodes synchronized onto two different common solutions related via the symmetry $\gamma$. Note that this is a more general definition than that presented in \cite{altafini2013consensus} where the scalar asymptotic solutions considered therein agree in modulus but differ in sign. In our case the two solutions $s$ and $s^\ast$ still share the same norm but are related by the more generic symmetry transformation $\gamma$.


\subsection{Main Result}

The following result provides a sufficient condition for network \eqref{eq:network_diffusive_nonlinear} to achieve a $\gamma$-bipartite synchronization pattern.
%
\begin{theorem}
\label{thm:pattern_synch}
Network  \eqref{eq:network_diffusive_nonlinear} achieves a $\gamma$-bipartite synchronization pattern if:
\begin{description}
\item [{\bf H1}] the intrinsic node dynamics $f$ is $\gamma$-equivariant, with $\gamma\in\mathbb{O}(n)$;
\item [{\bf H2}] $g_{ij}$ is chosen as follows:
$$
g_{ij}\left(x_j\right) := 
\begin{cases}
x_j, & i,j \in \mathcal{G}  \mbox{\, or\, } i,j \in \mathcal{G^\ast}\\
\gamma\, x_j, & i \in  \mathcal{G} \mbox{\, and \, }  j \in \mathcal{G}^\ast\\
\gamma^{T} x_j,  & i \in  \mathcal{G}^\ast \mbox{\, and \,} j \in \mathcal{G}
\end{cases}
$$
\item [{\bf H3}] the associated auxiliary network \eqref{eq:network_block} synchronizes.
\end{description}
\end{theorem}

\begin{proof}
Without loss of generality, let us consider the first $\ell$ nodes belonging to the subset $\mathcal{G}$, that is $\mathcal{G}=\{1,\dots,\ell\}$, and the remaining nodes to $\mathcal{G}^\ast$, that is $\mathcal{G}^\ast=\{\ell+1,\dots,N\}$. Hypothesis {\bf H2} implies that the  dynamics of network \eqref{eq:network_diffusive_nonlinear} can be written as follows.

%
%
\begin{eqnarray*}
\dot{x}_i= f(t,x_i)-k\Bigg[ l_{ii}x_i + \sum_{\substack{j=1\\ j\neq i}}^\ell l_{ij}x_j + \sum_{j=\ell+1}^N l_{ij}\gamma\,x_j \Bigg], \\ \quad \mbox{if }i\in\mathcal{G}; \\
\dot{x}_i= f(t,x_i)-k\Bigg[ l_{ii}x_i + \sum_{j=1}^\ell l_{ij}\gamma^T x_j + \sum_{\substack{j=\ell+1\\ j\neq i}}^N l_{ij}x_j \Bigg], \\ \quad \mbox{if }i\in\mathcal{G}^\ast,
\end{eqnarray*}
where $l_{ij}$ are the elements of the Laplacian matrix. Now, let $D$ be the $nN\times nN$ block-diagonal matrix having on its main block-diagonal
\begin{equation}
\label{eq:sigma_bi}
\sigma_i=
\begin{cases}
I_n & \mbox{if node $i$ belongs to } \mathcal{G}\\
\gamma & \mbox{if node $i$ belongs to } \mathcal{G}^\ast\\
\end{cases}
\end{equation}
Then the above dynamics can be rewritten in compact form as  $\dot{X}=F(t,X)-k\, D^{T}(L\otimes I_n)DX$ (recall that $D^T = D^{-1}$). Now, let $Z = DX$, we have:
\begin{equation*}
\begin{split}
\dot Z = & DF(t,X)-k\,DD^T(L\otimes I_n)DX=\\
=& F(t,DX)-k\,(L\otimes I_n)DX,
\end{split}
\end{equation*}
where we used {\bf H1} and Lemma \ref{lemma:F_equivariant}. Therefore, in the new state variables $Z$, the network dynamics can be recast as
\begin{equation}
\label{eq:network_Z}
\dot{Z}=F(t,Z)-k(L\otimes I_n)Z,
\end{equation}
that  has the same form as the auxiliary network \eqref{eq:network_block}. Now, from hypothesis {\bf H3}, since the auxiliary network synchronizes, then so does network \eqref{eq:network_Z} which shares the same network dynamics. Therefore, there exists some $\dot s = f(t,s)$ such that $\lim_{t\rightarrow + \infty} \abs{z_i(t)-s(t)} = 0, \forall i = 1,\ldots, N$.

Now, $X=D^{T}Z$ yields
\begin{equation*}
\begin{split}
\lim_{t\to +\infty} x_i(t) 
=
\begin{cases}
I_n\, z_i(t)=s(t), & \mbox{if } i\in \mathcal{G};\\
\gamma^{T} z_i(t)=\gamma^{T} s(t), & \mbox{if } i\in \mathcal{G}^\ast.
\end{cases}
\end{split}
\end{equation*}
Finally, from Definition \ref{def:bipartite_sync} we know that $s(t) = \gamma s^\ast(t)$, thus $\gamma^Ts(t) = \gamma^T\gamma s^\ast (t)= s^\ast (t)$, since $\gamma^T = \gamma^{-1}$ from {\bf H1}.
\end{proof}
%
%
\begin{remark}
In the proof of Theorem \ref{thm:pattern_synch}, we used a transformation matrix $D$ which is a generalization of the one used in \cite{altafini2013consensus}, where only the set of \emph{gauge transformations} $D=\mathrm{diag}\{\pm 1,\dots,\pm 1 \}$ was considered for a network of (scalar) integrators. Indeed, when $n=1$ the orthogonal group, and therefore every possible $\gamma$, is $\mathbb{O}(1)=\{1,-1\}$.
\end{remark}
%
Note that Theorem \ref{thm:pattern_synch} reduces the problem of proving convergence to a $\gamma$-bipartite synchronization pattern  to that of ensuring synchronization of an auxiliary network which is diffusively coupled. In general, this latter problem is much simpler to solve than the former (see e.g. \cite{Rus_diB_09b,deL_diB_Rus_11,Pec_Car_98}).
%
%
%
%
Also, notice that our result is somewhat complementary to the one given in \cite{pogromsky200265} where the stability is analyzed of synchronization patterns arising from symmetries of the network structure.  Here, instead, we use symmetries in the nodes' dynamics to induce synchronization patterns in the network.
%


%
%

\section{Application to linear systems}
\label{sec:consensus}
%

Consider a set of $N>1$ LTI agents described by
\begin{equation}
\label{eq:lti}
\dot{x}_i=Ax_i+Bu_i
\end{equation}
where $i=1,\dots,N$, $x_i\in\mathbb{R}^n$, $A\in\mathbb{R}^{n\times n}$, $B\in\mathbb{R}^{n\times m}$, and assume they are networked through the interconnection protocol $u_i\in\mathbb{R}^m$ given by
\begin{equation}
\label{eq:lti_u}
u_i=K\, \sum_{j=1}^N a_{ij} (g_{ij}(x_j)-x_i)
\end{equation}
where $K\in\mathbb{R}^{m\times n}$ is the control gain matrix, $a_{ij}\in\{0,1\}$ and $g_{ij}$ is the coupling function defined as before. Substituting \eqref{eq:lti_u} into \eqref{eq:lti}, we obtain
\begin{equation}
\label{eq:lti_cl}
\dot{x}_i=Ax_i+BK\, \sum_{j=1}^N a_{ij} (g_{ij}(x_j)-x_i)
\end{equation}
for $i=1,\dots,N$. As noted in Section \ref{sec:bipartite_synch}, if we select the coupling functions as $g_{ij}(x)=x$, then we obtain a diffusively coupled network that can be written in compact form as \cite{zhang2016bipartite} 
\begin{equation}
\label{eq:lti_aux}
\dot{X}=(I_N\otimes A)X-(L\otimes BK)X
\end{equation}
where $L$ is the Laplacian matrix. Again, we will refer to network \eqref{eq:lti_aux} as the \emph{auxiliary network} associated to \eqref{eq:lti_cl}.

\begin{corollary}
\label{thm:lti}
A $\gamma$-bipartite consensus pattern arises for network \eqref{eq:lti_cl}  if Theorem \ref{thm:pattern_synch} holds for some  $\gamma \in \mathbb{O}(n)$ such that $A\,\gamma=\gamma\, A$.
\end{corollary}
\begin{proof}
The proof follows the same steps as in Theorem \ref{thm:pattern_synch}. In particular, using hypothesis \textbf{H2}, we can rewrite  \eqref{eq:lti_cl} as
\begin{equation*}
\dot{X}=(I_N\otimes A)X-D^T(L\otimes BK)DX
\end{equation*}
where $D$ is defined as in \eqref{eq:matrix_D} and \eqref{eq:sigma_bi}. Now, note that when $f(x)=Ax$, $f$ being $\gamma$-equivariant (\textbf{H1}) means that the matrices $A$ and $\gamma$ commute. 
%
Moreover, note that $D(I_N\otimes A)=(I_N\otimes A)D$, since $D$ and $(I_N\otimes A)$ are block diagonal matrices whose respective diagonal blocks commute with each other. Therefore, taking $Z=DX$ we obtain
\begin{equation*}
\dot Z  = (I_N\otimes A)Z-(L\otimes BK)Z,
\end{equation*}
that has the same form of the auxiliary network \eqref{eq:lti_aux}. From \textbf{H3}, this latter network achieves consensus and therefore, as in the proof of Theorem \ref{thm:pattern_synch}, we can conclude that network \eqref{eq:lti_cl} achieves $\gamma$-bipartite consensus.
\end{proof}
\begin{remark}
In \cite{zhang2016bipartite} the authors studied bipartite consensus in diffusive networks,
which, in the context of our results, corresponds to the case where $\gamma=-I_n$. Note that the matrix $-I_n$ commutes with every square matrix $A$ and therefore the result of \cite{zhang2016bipartite} is a special case of Corollary \ref{thm:lti}.
\end{remark}

As a specific application, we consider a connected undirected network of $N>1$ $n$-dimensional integrators $y_i^{(n)}=u_i$, where $y_i$ is a scalar and $y_i^{(h)}$ denotes $d^h y_i/dt^h$. The model can be written in compact form as \eqref{eq:lti_cl}, where $x_i:=[y_i,\,\dot{y}_i,\dots,\, y_i^{(n-1)} ]^T$,
\begin{equation}
\label{eq:integrators_AB}
A=
\begin{bmatrix}
0 & 1 & 0 & \dots & 0 & 0\\
0 & 0 & 1 & \dots & 0 & 0\\
\vdots & \vdots & \vdots & \ddots & \vdots & \vdots\\
0 & 0 & 0 & \dots & 1 & 0\\
0 & 0 & 0 & \dots & 0 & 1\\
0 & 0 & 0 & \dots & 0 & 0\\
\end{bmatrix};
\quad
B=
\begin{bmatrix}
0\\
0\\
\vdots\\
0\\
0\\
1
\end{bmatrix}
\end{equation}
and $K\in\mathbb{R}^{1\times n}$ is the control gain matrix.
Given the special structure of $A$, Corollary \ref{thm:lti} implies that the only possible $\gamma$-bipartite consensus pattern for \eqref{eq:integrators_AB} is with $\gamma=-I_n$.

\section{Multipartite synchronization}
\label{sec:multipartite}
We present a generalization of the results of Section \ref{sec:bipartite_synch} to the case of ODEs having more than one symmetry. Let $\mathcal{G}_N := \left\{1, \ldots, N \right\}$ be the set of nodes and let $\mathcal{G}_1, \dots, \mathcal{G}_r$ be $r\geq 2$ non-empty subsets forming a partition for $\mathcal{G}_N$, that is $\mathcal{G}_i\, \cap\, \mathcal{G}_j = \left\{ \emptyset \right\}$, for all $i,j$, with $i\neq j$, and $\bigcup_{i=1}^r\,\mathcal{G}_i = \mathcal{G}_N$.

\begin{definition}
Consider network \eqref{eq:network_diffusive_nonlinear}  and let $\{\gamma_1,\dots,\gamma_r\}\in\mathbb{O}(n)$. We say that \eqref{eq:network_diffusive_nonlinear} achieves a \emph{$\Gamma$-multipartite synchronization pattern} if
\begin{equation*}
\begin{array}{*{20}l}
\lim_{t\to +\infty} |x_i(t)-s_1(t)|=0, & \quad \forall i\in \mathcal{G}_1;\\
\qquad \quad \vdots  \\
\lim_{t\to +\infty} |x_i(t)-s_r(t)|=0, & \quad \forall i\in \mathcal{G}_r,
\end{array}
\end{equation*}
where
\begin{equation*}
\begin{array}{*{20}l}
s_1(t)=\gamma_1\, s_1(t)=I_n\, s_1(t)\\
\qquad \quad \vdots\\
s_1(t)=\gamma_r\, s_r(t).
\end{array}
\end{equation*}
%
\end{definition}
\vspace{0.2cm}

\begin{theorem}
\label{thm:multipartite}
Network \eqref{eq:network_diffusive_nonlinear} achieves a $\Gamma$-multipartite synchronization pattern if:
\begin{description}
\item [{\bf H1}]  the intrinsic node dynamics $f$ is $\Gamma$-equivariant, and there exist $r$ symmetries $\{\gamma_1,\dots, \gamma_r\}\in\Gamma$;
\item [{\bf H2}] $g_{ij}$ is defined as:
(i) $g_{ij}(x_j) = \gamma_h^T\gamma_k x_j$ if $i \in \mathcal{G}_h$ and $j \in \mathcal{G}_k$; (ii) $g_{ij}(x_j)=x_j$ if $i$ and $j$ belong to the same cluster;
\item [{\bf H3}] the associated auxiliary network \eqref{eq:network_block} synchronizes.
\end{description}
\end{theorem}
%
%
%
%
\begin{proof}
Without loss of generality, relabel the network nodes such that the first $\ell_1$ nodes belong to $\mathcal{G}_1$, i.e. $\mathcal{G}_1=\{1,\dots, \ell_1\}$, then the other $\ell_2-\ell_1$ nodes belong to $\mathcal{G}_2$, i.e. $\mathcal{G}_2=\{\ell_1+1,\dots, \ell_2\}$, and so on until $\mathcal{G}_{r}=\{\ell_{r-1}+1,\dots,\ell_r\}$, with $\ell_r=N$. From hypothesis {\bf H2} the network dynamics \eqref{eq:network_diffusive_nonlinear} can then be written as
%
%
%
\begin{equation*}
\begin{split}
\dot{x}_i & =  f(t,x_i)  -k\Bigg[ l_{ii}x_i + \sum_{\substack{j=1\\ j\neq i}}^{\ell_1} l_{ij}\, \gamma_h^T \gamma_1 \,x_j\\
& + \sum_{\substack{j=\ell_1+1\\ j\neq i}}^{\ell_2} l_{ij}\,\gamma_h^T\gamma_2\,x_j + \dots  
 + \sum_{\substack{j=\ell_{r-1}+1\\ j\neq i}}^{\ell_r} l_{ij}\,\gamma_h^T\gamma_r\,x_j \Bigg], 
\end{split}
\end{equation*}
for any $i\in\mathcal{G}_h$ and $h\in\{1,\dots,r\}$,
where $l_{ij}$ are the elements of the Laplacian matrix. Now, let $D$ be the $nN\times nN$ block-diagonal matrix defined in \eqref{eq:matrix_D} having on its main diagonal the blocks $\sigma_i= \gamma_h$ if node $i$ belongs to $\mathcal{G}_h$, with $h\in\{1,\dots,r\}$. Then, the above dynamics can be rewritten as (recall that $D^T = D^{-1}$):
\begin{equation}
\label{eq:network_X_multi}
\dot{X}=F(t,X)-k\,D^{T}(L\otimes I_n)DX.
\end{equation}
Now, let $Z = DX$, we have:
\begin{equation*}
\begin{split}
\dot Z =&  DF(t,X)-k\,DD^T(L\otimes I_n)DX=\\
=& F(t,Z)-k\,(L\otimes I_n)Z,
\end{split}
\end{equation*}
where we used {\bf H1} and Lemma \ref{lemma:F_equivariant}. Now, note that in the new state variables the network dynamics can recast as
\begin{equation}
\label{eq:network_Z_multi}
\dot{Z}=F(t,Z)-k\,(L\otimes I_n)Z,
\end{equation}
that has the same form as the auxiliary network \eqref{eq:network_block}.

Since by hypothesis {\bf H3}, the auxiliary network synchronizes, then also does network \eqref{eq:network_Z_multi}. Therefore, there exists some $\dot s_1 = f(t,s_1)$ such that, $\forall i = 1,\ldots, N$: $
\lim_{t\rightarrow + \infty} \abs{z_i(t)-s_1(t)} = 0$, $\forall i$. Since $X=D^{T}Z$, we finally have that
\begin{equation*}
\begin{split}
\lim_{t\to +\infty} x_i(t) & = 
\begin{cases}
\gamma_1^T\, s_1(t)= I_n\, s_1(t)=s_1(t), & \mbox{if } i\in \mathcal{G}_1;\\
\gamma_2^{T} s_1(t)=s_2(t), & \mbox{if } i\in \mathcal{G}_2;\\
\qquad \vdots\\
\gamma_r^{T} s_1(t)=s_r(t), & \mbox{if } i\in \mathcal{G}_r.
\end{cases}
\end{split}
\end{equation*}
\end{proof}
\section{A design methodology}
\label{sec:design}
%
We  show how the results of this paper can be used to design distributed control strategies ensuring that a generic network of interest attains a desired synchronization/consensus pattern. The methodology considers a {\em local nonlinear controller}, $v_i(x_i)$, at the node level inducing a symmetry in its closed-loop vector field and a {\em communication protocol} that exploits this symmetry to attain the desired synchronization pattern.  The resulting closed loop network dynamics takes the form
\begin{equation*}
\label{eq:network_algorithm}
\dot x_i = f\left(t,x_i\right) + v_i(x_i)+ k\, \sum_{j=1}^N a_{ij}\, \left(g_{ij}\left( x_j\right) - x_i\right).
\end{equation*}
%
The control task in this case is to ensure that a desired $\Gamma$-multipartite pattern is achieved by the network. Typical target patterns can include
\begin{itemize}
\item
\emph{anti-synchronization}, where nodes belonging to different clusters synchronize onto two different trajectories, say $s_1(t)$ and $s_2(t)$, with $s_2(t) = - s_1(t)$. In this case, the synchronous evolution of the two clusters are related via the symmetry  $\gamma=-I_n$, i.e. $s_1(t) = \gamma s_2(t)$;
\item
\emph{partial anti-synchronization}, where nodes belonging to different clusters have a subset of state variables which are synchronized and a subset of state variables which are anti-synchronized. In this case, the synchronous evolution of the different clusters are related via the symmetry $\gamma=\mathrm{diag}\{\pm 1,\dots, \pm 1 \}$;
\item
\emph{phase-shift synchronization}, also called discrete rotating wave, where nodes belonging to different clusters synchronize on the same $T$-periodic solution $s(t)$ but with a different phase-shift $\theta\in[0,1)$. In this case, $\gamma x_i(t)=x_j(t+\theta T)$.
\end{itemize}
%

Our procedure to ensure that a desired $\Gamma$-multipartite pattern is achieved consists of the following steps:
\begin{enumerate}
\item
Choose the symmetry group $\Gamma=\{\gamma_1,\dots,\gamma_r\}$ associated to the desired synchronization pattern that the network should achieve;
\item Determine the desired partition $\mathcal{G}_1,\ldots,\mathcal{G}_r$ identifying nodes in each clusters that should exhibit a different synchronous solution;
\item For all the nodes belonging to the cluster $\mathcal{G}_i$, check if $f(t,x)$ is $\gamma_i$-equivariant, with $\gamma_i\in\Gamma$. If this condition is verified, then set $v_i(x)=0$. Otherwise, design the local nonlinear control input such that the closed-loop vector field
$\widehat{f}_i(t,x):=f(t,x)+v_i(x)$ is $\gamma_i$-equivariant. \textcolor{black}{For example, if the desired pattern is anti-synchronization, then the local controller $v_i(x)$ has to be designed such that the vector field $\widehat{f}_i(t,x)$ is an odd function (see Section \ref{sec:antiynch}). Analogously, in the case where phase-shift is the desired pattern, then $v_i(x)$ has to be designed such that $\widehat{f}_i(t,x)$ commutes with a matrix belonging to the group $\mathbb{SO}_2$ (see Section \ref{sec:example_multi})};
\item Design the communication protocols in accordance to {\bf H2} of Theorem~\ref{thm:multipartite} \textcolor{black}{and set the coupling gain $k$ so that the corresponding auxiliary network \eqref{eq:network_block} synchronizes, in accordance to {\bf H3}.}

\end{enumerate}

\section{Examples}
\subsection{Anti-synchronization of FitzHugh-Nagumo oscillators}\label{sec:antiynch}
We address the problem of generating an anti-synchronization pattern for a network of FitzHugh-Nagumo (FN) oscillators \cite{Fit_55}.
The network is described by
\begin{equation}\label{eqn:FN}
\begin{array}{*{20}l}
\dot v_i = c\left(v_i+w_i -\frac{1}{3}v_i^3\right) + k\sum_{j=1}^N a_{ij}\,(v_j-v_i),\\
\dot w_i = -\frac{1}{c}\left(v_i+bw_i\right)+ k\sum_{j=1}^N a_{ij}\,(w_j-w_i),\\
\end{array}
\end{equation}
where $v_i$ and $w_i$ are the membrane potential and the recovery variable for the $i$-th FN oscillator ($i=1,\ldots,N$).
In terms of the formalism introduced in Definition \ref{def:bipartite_sync}, anti-synchronization will correspond to the case where $s(t) = - s^\ast(t)$ so that $\gamma = -I_2$.
Now, let $x=[v,w]^T$. Then, $f(t,x)=\left[c\left(v+w -\frac{1}{3}v^3 \right), -\frac{1}{c}(v+bw)\right]^T$ fulfills hypothesis {\bf H1} of Theorem \ref{thm:pattern_synch}. 
%
Consider now the network structure shown in Figure \ref{fig:net_topologies} (left panel) and its partition (right panel), obtained by dividing nodes into the two clusters $\mathcal{G}=\{1,3\}$ and $\mathcal{G}^\ast=\{2,4,5\}$. The nodes' dynamics can then be written according to \eqref{eq:network_diffusive_nonlinear} as
\begin{equation*}
\label{eqn:FN_network}
\begin{array}{*{20}l}
\dot x_1 = f\left(x_1\right)+k \left(g_{12}\left(x_2\right)+g_{13}\left(x_3\right)-2x_1\right)\\
\dot x_2 = f\left(x_2\right)+k \left(g_{21}\left(x_1\right)-x_2\right)\\
\dot x_3 = f\left(x_3\right)+k \left(g_{31}\left(x_1\right)+g_{34}\left(x_4\right)+g_{35}\left(x_5\right)-3x_3\right)\\
\dot x_4 = f\left(x_4\right)+k \left(g_{43}\left(x_3\right)-x_4\right)\\
\dot x_5 = f\left(x_5\right)+k \left(g_{53}\left(x_3\right)-x_5\right)
\end{array}
\end{equation*}

It is well know from the literature that the auxiliary network \eqref{eq:network_block} associated to the above dynamics synchronizes if $k$ is sufficiently large \cite{strogatz2001exploring}. Therefore, by choosing the coupling gain $k$ sufficiently high, hypothesis {\bf H3} of Theorem \ref{thm:pattern_synch} will also be fulfilled. Finally, {\bf H2} of Theorem \ref{thm:pattern_synch} is fulfilled by  choosing the coupling functions $g_{ij}$ as: (i) $g_{ij}\left(x_j\right)= x_j$, if $i, j \in \mathcal{G}$, or if $i, j \in \mathcal{G^\ast}$; (ii) $g_{ij}\left(x_j\right)= -x_j$, if $i \in  \mathcal{G}$ and $j \in \mathcal{G}^\ast$ or if  $i \in  \mathcal{G}^\ast$ and $j \in \mathcal{G}$. With this choice of the functions $g_{ij}$'s, in accordance to Theorem \ref{thm:pattern_synch}, anti-synchronization is attained with nodes $1$ and $3$ converging onto the same trajectory $s(t)$, while  nodes $2$, $4$ and $5$  onto $s^\ast=-s(t)$ (see Figure \ref{fig:pattern_FN}).
\begin{figure}[t]
\centering
  {\includegraphics[width=3.5cm]{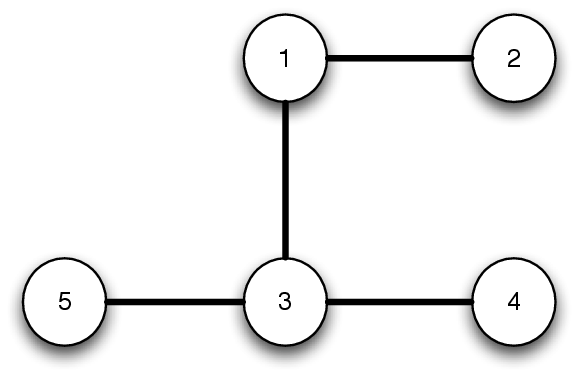}
  \includegraphics[width=3.5cm]{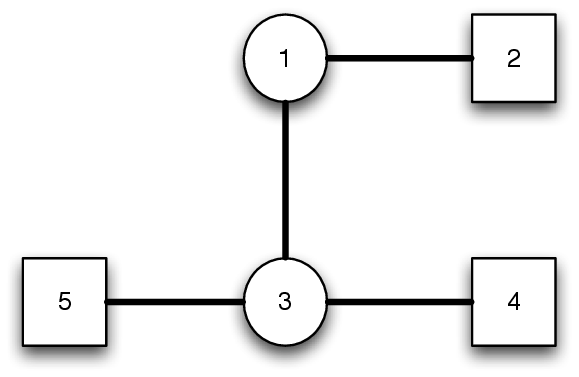}}
  \caption{Left panel: a network of diffusively coupled identical oscillators. Right panel: an arbitrary network partition. Nodes 1 and 3 belong to cluster $\mathcal{G}$, nodes 2, 4 and 5 belong to cluster $\mathcal{G}^\ast$.}
  \label{fig:net_topologies}  
\end{figure}
\begin{figure}[t]
\begin{center}
  \includegraphics[width=7.5cm]{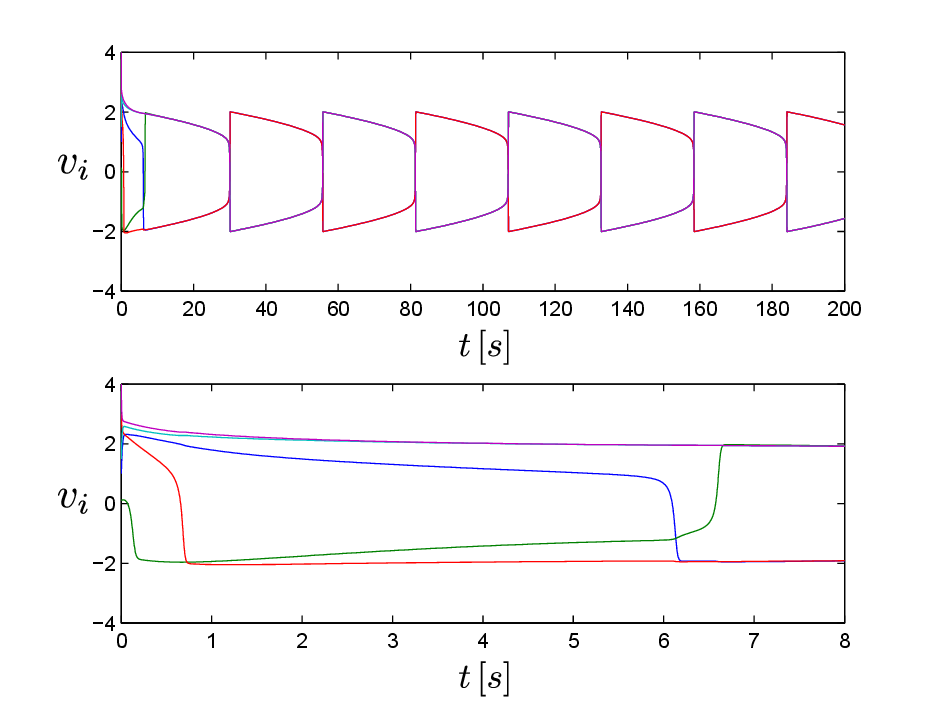}
  \caption{Top panel: time behavior of \eqref{eqn:FN}, with $k=1$. Note that two clusters of anti-synchronized nodes arise. Bottom panel: transient behavior of network nodes. Initial conditions are randomly taken from the standard distribution.}
  \label{fig:pattern_FN}
  \end{center}
\end{figure}
%
%
%
\subsection{Generating wave patterns}
\label{sec:example_multi}
%
%
\begin{figure}[t]
\begin{center}
\includegraphics[width=7cm]{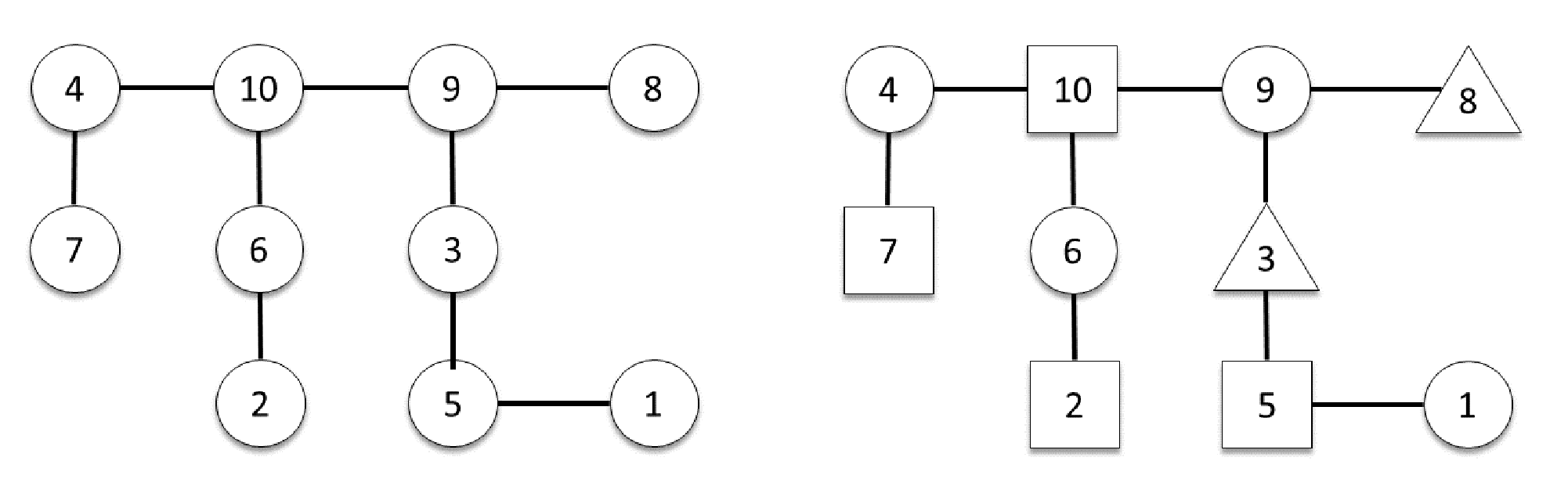}
\caption{Left panel: a network of diffusively coupled harmonic oscillators \eqref{eq:harmonic_osc}. Right panel: a  network partition:  $\mathcal{G}_1$ (squares), $\mathcal{G}_2$ (circles), $\mathcal{G}_3$ (triangles).}
\label{fig:ex_multi_net}
\end{center}
\end{figure}
In order to illustrate the application of Theorem \ref{thm:multipartite}, we consider the problem of generating a {\em discrete rotating wave} pattern for a network of harmonic oscillators.
%
%
%
Specifically, we consider a network of $N=10$ identical harmonic oscillators with topology as in Figure \ref{fig:ex_multi_net} (left panel). The harmonic oscillator dynamics is described by
\begin{equation}
\label{eq:harmonic_osc}
\dot{x}=Ax=
\begin{bmatrix}
0 & -\omega\\
\omega & 0
\end{bmatrix}
x .
\end{equation}
The symmetries of \eqref{eq:harmonic_osc} are those described by rotations by an angle $\phi\in[0,2\pi)$. That is, $\gamma$ belongs to the special orthonormal group $\mathbb{SO}(2)$ or, in matrix form, $\gamma=
\begin{bmatrix}
\cos\,\phi & -\sin\,\phi\\
\sin\,\phi & \cos\,\phi
\end{bmatrix}.$ It is important to note that a set of weakly coupled nonlinear oscillators can be transformed via the so-called \emph{phase reduction} \cite{kuramoto1984chemical} into a new set of ODEs that is equivariant with respect to the circle group $\mathbb{S}^1$, which is isomorphic to $\mathbb{SO}(2)$. To satisfy hypothesis \textbf{H1} of Theorem \ref{thm:multipartite}, consider, for example, three symmetries $\gamma_1$, $\gamma_2$ and $\gamma_3$ associated to rotations by $\phi_1=0^\circ$, $\phi_2=120^\circ$ and $\phi_3=240^\circ$, respectively, and consider the network nodes partitioned into $\mathcal{G}_1=\{2,5,7,10\}$, $\mathcal{G}_2=\{1,4,6,9\}$ and $\mathcal{G}_3=\{3,8\}$ associated to the respective symmetries, as reported in Figure \ref{fig:ex_multi_net} (right panel). Applying the coupling functions in accordance to \textbf{H2} of Theorem \ref{thm:multipartite}, the network dynamics is described by \eqref{eq:network_X_multi} where $F(X)=(I_N\otimes A)X$ and the matrix $D$ is the diagonal matrix $D=\mathrm{diag}\,\left\{\gamma_2,\gamma_1,\gamma_3,\gamma_2,\gamma_1,\gamma_2,\gamma_1,\gamma_3,\gamma_2,\gamma_1\right\}$. Furthermore, following Theorem $5.1$ in \cite{diB_Liu_Rus_14}, it can be shown that the auxiliary network \eqref{eq:network_Z_multi} synchronizes for any $k>0$, and therefore all hypotheses of Theorem \ref{thm:multipartite} are verified. As expected (see Figure \ref{fig:ex_multi}) the nodes belonging to the same cluster synchronize with each others, with a phase delay of $120^\circ$ between the clusters. 

\begin{figure}[t]
\begin{center}
\includegraphics[width=6.5cm]{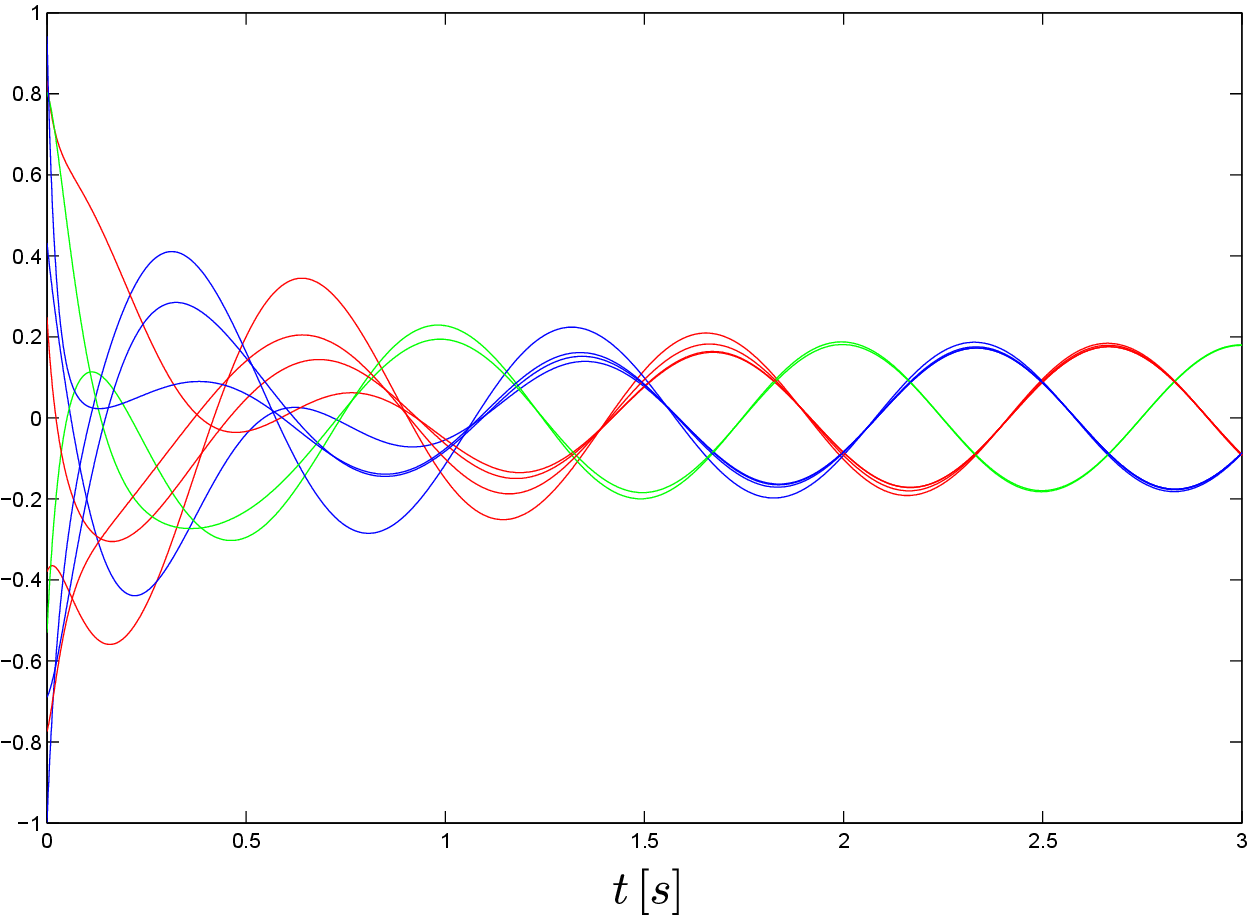}
\caption{Time behavior of $x_1$ of the nodes in $\mathcal{G}_1$ (blue), $\mathcal{G}_2$ (red) and $\mathcal{G}_3$ (green), with $\omega=1$ and $k=10$. Initial conditions are randomly taken from the uniform distribution on the unit circle.}
\label{fig:ex_multi}
\end{center}
\end{figure}


\section{Conclusions}
\label{sec:conclusions}

We showed that symmetries of the nodes' dynamics can be exploited to guarantee the onset of synchronization/consensus patterns in networks. After presenting a set of sufficient conditions ensuring emergence of both bipartite and multipartite synchronization/consensus patterns, we demonstrated the effectiveness of our methodology through a set of representative examples.
Our results generalize those presented in \cite{altafini2013consensus} to higher dimensional and nonlinear settings. Future work will be aimed at relaxing the assumptions of identical node dynamics and linear coupling. Also, the investigation of the link between our results and the theory of equivariant dynamics in \cite{golubitsky2015recent,josic2006network,whalen2015observability} deserves further attention.







\bibliographystyle{IEEEtran} 
\bibliography{refs}   


\end{document}